\documentclass[11pt]{article} 

\usepackage[utf8]{inputenc} 

\usepackage{amsmath,amsthm} 
\usepackage{amssymb} 
\usepackage{comments}
\usepackage{boxedminipage}

\usepackage{tikz}
\usetikzlibrary{arrows,matrix}
\usepackage{graphicx} 
\usepackage{proof}
\usepackage{bussproofs}
\title{A geometrical view of I/O logic \\ (Nr 432)}
\author{D. Gabbay, X. Parent and L. van der Torre}

\newtheorem{defi}{Definition}
\newtheorem{coro}[defi]{Corollary}

\newtheorem{exa}[defi]{Example}

\newtheorem{ass}[defi]{Assumption}
\newtheorem{conj}[defi]{Conjecture}

\newtheorem{rem}[defi]{Remark}

\newtheorem{theo}[defi]{Theorem}

\newtheorem{problem}{Problem}

\newtheorem{Obs}{Observation}
\newtheorem{Le}{Lemma}

\begin{document}
\maketitle

\section{Introduction}

A frequent belief about I/O logic is that it presupposes classical propositional
logic. This is clearly  a
misunderstanding. It is true that in their seminal paper Makinson and van der
Torre~\cite{mvt2000} found it more convenient to build I/O logic on
top of classical logic. One reason why is that it is a simple and
well-defined framework. However in the study \cite{Parent2014} we have shown  the base logic can also be 
intuitionistic propositional logic. Three of the four standard I/O operations defined by
Makinson and van der Torre are covered. 
It is shown that the axiomatic characterizations available for them hold also for intuitionistic logic. Of course, some important changes must be made to the semantics. In particular, the role played by maximal consistent set is now played by the notion of saturated set. In \cite{imla17} the second author goes one step further, and  shows that the intuitionistic analog of the so-called ``basic" I/O operation $out_2$ has a modal translation into de Paiva's system of constructive modal logic CK \cite{depaiva}.

In this note we  investigate if one can use an algebraic setting  instead of a logic setting.  
The main
result obtained so far concerns the simple-minded output operation
$out_1$. First, we show how to define it within this new
set-up. Second, we argue that all that is needed to get
a complete axiomatization is to be in a lattice. In terms of
connectives, this means that only conjunction is required. This is
truly a generalization of the initial set-up. For classical logic
corresponds to the special case of a lattice that is a boolean
algebra. This one could equally be an Heyting algebra.


The interesting thing about algebras is that they allow us to approach
I/O logic from a geometrical point of view. An algebra is mainly about
moving from points to points along a relation $\geq$. The travelling
always goes in the upwards direction.  We can think of the set $G$ of
generators used in I/O logic as ``jump" points or bridges. A pair
$(a,x)\in G$ is an instruction to deviate from a path that would be
normally taken. Once we have reached node $a$, instead of continuing
up we jump to an unrelated node $x$ and continue your journey upwards.

There is a connection with the Lindhal algebraic approach to normative
systems (see beginning of section~\ref{out1}). 


\section{Background}

We start with some background.
\begin{defi}[Lattice]\label{lat} A lattice is a structure $\langle \mathcal{L}, \geq,\wedge, 0,1\rangle$ where
\begin{itemize}
\item [i)] $\mathcal{L}$ is a set of points, $a$, $b$, ...
\item  [ii)]$\geq\subseteq \mathcal{L}\times\mathcal{L}$ is reflexive and transitive
\item  [iii)] $0$ and $1$ are the bottom and top element respectively 
\item  [iv)] any two elements $a$ and $b$ have 
a greatest lower bound or infimum or meet (denoted by $a\wedge b$) defined by 
\begin{itemize}
\item [\;] $a\wedge b = \inf (a,b)= \max\{z\mid z\leq a \,\&\, z\leq b \}$
\end{itemize}
\end{itemize}
\end{defi}
For $``\inf (a,b)"$ read ``the infimum of  $\{a,b\}$". The above definition is
for pairs, but can easily be generalized to sets of arbitrary size:
$$
\inf A = \max \{x : \forall y \;( y\in A \rightarrow x\leq y\}
$$
Note that the infimum of the emptyset is $1$. This is because, for 
$A=\emptyset$, the implication  $y\in A \rightarrow x\leq y$ holds for any
$x\in \mathcal{L}$, and thus $\inf A = \max \mathcal{L}$.



An ordered set that has properties iii) and iv) is usually called a
lattice.\footnote{Cf.~\cite[p.75]{cl00}.}



If $x\leq y$, we say that $x$ is below $y$ (or $y$ is above $x$).

The notation ``$\uparrow a$" will be for the upset of $a$, viz $\{x:
a\leq x\}.$ \\

The function $G(.)$ (which is a characteristic part of I/O logic)
takes a subset $A$ of $\mathcal{L}$ as argument, and delivers the set
of points to which we can ``jump" starting from these elements:
\begin{flalign}
G(A)=\{x\mid (a,x)\in G \mbox{ for some } a\in A\} \tag{Jump}
\end{flalign}

Before redefining the input/output operations using this apparatus, we
need to clarify what the counterpart of the consequence relation $Cn$
is within this set-up. When the premisses set is a single point $a$ in
the algebra, it is okay to identify $Cn$ with $\geq$, and say that the
consequences of $a$ are just all the points above it. But this leaves
unanswered the question of what the consequences of a \emph{set} of
formulae (or points) are. Our suggestion is to identify $Cn(A)$ with
the upset of the infimum of $A$, $\uparrow \inf A$. For instance, if
$A=\{a,b\}$, the analogue of the consequence set of $A$ is $\uparrow
\inf \{a,b\}=\uparrow a\wedge b$.

Note that a set $A$ bounded from above may not have an infimum. This
is because the ordering is a partial order. So it may happen that,
amongst the lower bounds of $A$, none is greater than all the
others. This is illustrated with Figure~\ref{noinfimum}. Put
$A=\{a_3,a_4\}$. The lower bounds are $a_1$, $a_2$, and $0$. None is
maximal, and so $\inf A=\emptyset$ even though $A$ is bounded from
above by $0$. This possibility is ruled out by condition iv) in Definition \ref{lat}.

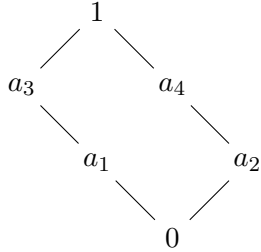
\begin{figure}[ht]
    \caption{$\mathcal{L} = \{1,0,a_1,a_2,a_3,a_4\}$} \label{noinfimum}
    \begin{center}
    \begin{tikzpicture}
        \tikzstyle{every node} = [rectangle]
        \node (s1) at (0,0) {$0$};
        \node (s2) at (-1,1) {$a_1$};
        \node (s3) at (1,1) {$a_2$};
        \node (s4) at (-2,2) {$a_3$};
        \node (s6) at (0,2) {$a_4$};
        \node (s12) at (-1,3) {$1$};
        \foreach \from/\to in {s1/s2, s1/s3, s2/s4, s3/s6, s4/s12, s6/s12}
            \draw[-] (\from) -- (\to);
    \end{tikzpicture}
    \end{center}
\end{figure}

\vspace{0,2cm}
Throughout the remainder of this note we work with a finite $\mathcal{L}$ (and a finite $G$).

We record for future reference the following facts about
infimum. Intuitively, lemma~\ref{factaboutinf} says that if $A$ is a
subset of $B$, then the infimum of $B$ is necessarily below the
infimum of $A$.
\begin{Le} \label{factaboutinf} If $A\subseteq B$, then $\inf B\leq\inf A$.
\end{Le}
 \begin{proof} Assume $A\subseteq B$. We have $\inf (A)\geq x$ for all $x$ that are below each and every element of $A$. Since $A\subseteq B$, $\inf B$ is below each and every element of $A$. Therefore
 $\inf B\leq\inf A$ as required.
\end{proof}
\begin{Le} \label{fact2aboutinf} If  $x\leq a_1$,..., and
$x\leq a_n$, then $x\leq \inf(a_1,...,a_n)$.
\end{Le}
 \begin{proof} Assume that $x$ is below all the $a_i$'s. 
By definition, $\inf(a_1,...,a_n)$ is above any point $y$ that is
below all the $a_i$'s. Therefore, $\inf(a_1,...,a_n)$ is above $x$.
\end{proof}
We are now ready to show that the analogues of the Tarskian properties
for $Cn$ hold for $\uparrow \inf$; hence the idea of using one for the
other.
\begin{Obs}[Tarskian properties] \label{properties} We have
\begin{flalign}
i) \;& A\subseteq \uparrow \inf A \tag{inclusion} \\
ii) & \mbox{ If } A\subseteq B  \mbox{ then } \uparrow \inf A \subseteq\uparrow \inf B\tag{monotony} \\
iii)  \;&  \uparrow \inf  \uparrow \inf A = \uparrow \inf A\tag{idempotence} 
\end{flalign}
\end{Obs}
 \begin{proof} For inclusion, let $a\in A$. 
By definition of infimum, $\inf A\leq a$ since $a\in A$. So
$a\in\uparrow \inf A$.

For monotony, assume $A\subseteq B$, and let $a\in\uparrow\inf A$. The
former means $\inf A\leq a$.  By lemma~\ref{factaboutinf}, $\inf
B\leq\inf A$. By transitivity of $\leq$, $\inf B\leq a$, and thus
$a\in\uparrow\inf B$ as required.

The right-in-left inclusion of idempotence is just inclusion. For the
left-in-right inclusion, let $x\in \uparrow \inf \uparrow \inf A$. Since $\mathcal{L}$ is finite
, $x\in \uparrow \inf (a_1, ..., a_n)$ where
$\uparrow \inf A = \{a_1, ..., a_n\}$. So 
$\inf (a_1, ..., a_n)\leq x$ and  $\inf A\leq a_i$ for all $i=1, ...,n$. By
Lemma~\ref{fact2aboutinf}, $\inf A\leq \inf (a_1, ..., a_n)$. By
transitivity of $\leq$, $\inf A\leq x$ as required.
\end{proof}

\section{Geometrical simple-minded
\protect\footnotemark{}\protect\footnotetext{The name
  ``simple-minded'' is from Makinson and van der Torre~\cite{mvt2000}}}\label{out1}
\subsection{Analogy with the Lindahl-Odelstal algebraic approach}
We start wih the analogy between I/O logic and the Lindahl-Odelstal
algebraic approach to normative systems (see, e.g.,~\cite{LO2000}). It
puts things in perspective. They call their model the ``cis-model'',
where ``cis'' abbreviates conceptual implicative structures. There a
normative system is viewed as consisting of a system $B_1$ of
potential grounds (or descriptive conditions), and a system $B_2$ of
potential consequences (or normative effects). The set of norms are
the set $J$ of links or joinings from the system of grounds to the
system of consequences. The cis-model is illustrated in
Figure~\ref{csi}, where a norm is represented by an arrow from one
system to the other.
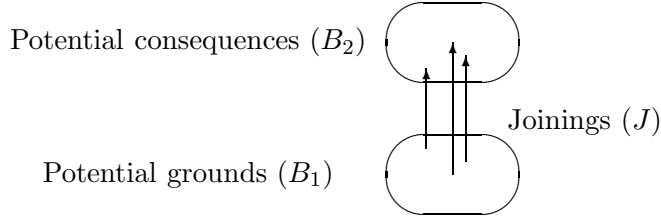
\begin{figure}[h]
\begin{center}
\begin{picture}(200,80)(00,00)

\put(100,50){\oval(50,30)}

\put(00,50){\makebox(0,0){Potential consequences ($B_2$)}}

\put(100,00){\oval(50,30)}
\put(105,5){\vector(0,1){40}}
\put(90,10){\vector(0,1){30}}
\put(100,0){\vector(0,1){50}}

\put(00,00){\makebox(0,0){Potential grounds ($B_1$)}}

\put(150,20){\makebox(0,0){Joinings ($J$)}}
\end{picture}
\end{center}
\caption[The cis-model]{The cis-model} \label{csi}
\end{figure} 
There is here an obvious connection with the simple-minded output operation $out_1$
described in~\cite{mvt2000}. Recall the
calculation of $out_1$ involves three main steps. First, input set $A$
is expanded to its classical closure $Cn(A)$. This corresponds to the
fact that grounds are ordered by an implication relation in the
cis-model. Next, the set obtained at step 1 is passed into a ``black
box", which delivers the corresponding deontic output. This is done by
taking the image of $Cn(A)$ under $G$. This is exactly what Lindahl
and Odelstal call ``joining''. Finally, the set of ouputs obtained at
step 2 is expanded to its classical closure again. This last step
corresponds to the fact that deontic consequences are also ordered by
an implicative relation. 

The above gives an idea of what a geometrical view of 
$out_1$ can be. It remains to fill in the details.

\subsection{Definition and result}\label{mainresult}

Below we reformulate the simple-minded output operation using lattices.
\begin{defi}[Simple-minded I/O]\label{simple-minded}
$$
out_1 (G,A)=\uparrow \inf(G(\uparrow\inf (A))
$$
\end{defi}
This definition mirrors the definition given  in~\cite{mvt2000}:
$$
out_1 (G,A)=Cn(G(Cn(A))
$$
There are three steps involved in the construction:
\begin{itemize}
\item Determine the upset of the infimum of $A$ $-$ this corresponds to the step $Cn(A)$;
\item Apply $G$ to the set obtained at step 1, viz. collect the nodes $x$
to which  a jump is allowed $-$ this corresponds to the step $G(Cn(A))$;
\item Take the upset of the infimum of the jump points obtained at step 2 
$-$ this corresponds to the step $Cn(G(Cn(A)))$;
\end{itemize}
In the particular case where the input set is singleton $a$, $\inf a = a$. This follows from reflexivity of
$\geq$. So definition~\ref{simple-minded} simplifies into:
$$
out_1 (G,a)=\uparrow \inf(G(\uparrow a))
$$ 
Below we give some examples where we relax the assumption $\mathcal{L}$ be finite. This is for illustration purpose only.

\begin{exa} Let $G=\{(p_1,p_2)\}$, where $p_1$ and $p_2$ are two propositional letters. $\mathcal{L}_1$ is the lattice of all the wffs generated from
$p_1$ and $p_2$. $out_1 (G,\{p_1\})$ is the set of wffs that can be
  proved from $p_2$.
\end{exa}

\begin{exa} Consider the set $\mathbb{N}$ of positive integers ordered by the
 divisibility relation $\mid$. The supremum and the infimum are given
 by the greatest common divisor and the least common multiple,
 respectively. $G$ can be any function from $\mathbb{N}$ to
 $\mathbb{N}$, e.g. the single-valued successor function $S$, which
 adds the value 1 to any number. $out_1 (G,x)$ is the set of all the
 numbers that the successor of $x$ divides. Turning upside down the
 lattice, $out_1 (G,x)$ is the set of all the divisors of the successor of
 $x$.
\end{exa}

Below we show that definition~\ref{simple-minded} validates the rules
that are characteristic of the simple minded operaton.
\begin{Obs}[SI and WO] If $out_1$ is taken as in definition~\ref{simple-minded}, then  $out_1$ validates
SI and WO thus 
reformulated \\

$
\begin{array}{ll}
  \textrm{(SI)} & \displaystyle\frac{(a,x) \hspace{1cm}
  b\leq a}{(b, x)} \\\\
  \textrm{(WO)} & \displaystyle\frac{(a,x) \hspace{1cm}
  x\leq y}{(a,y)} \\\\
\end{array}
$
\end{Obs}
\begin{proof} This basically follows from the transitivity of $\leq$.

For SI, assume $x\in out_1 (G,a)$ and $b\leq a$.  Let $x_1, ..., x_n$
be the bodies of the rules in $G$ that make $x\in out_1 (G,a)$
true. From definition $x_1, ..., x_n\in\uparrow a$. By transitivity of
$\leq$, $x_1, ..., x_n\in\uparrow b$, which suffices for $x\in out_1
(G,b)$.

The argument for WO is just symmetrical to that for SI.
Assume $x\in out_1 (G,a)$ and $x\leq y$.  Let $z$ be the infimum of
$\{b_1, ..., b_n\}$, where $b_1, ..., b_n$ are the heads of the rules
in $G$ that make $x\in out_1 (G,a)$ true. From definition,
$x\in\uparrow z$. By transitivity of $\geq$, $y\in\uparrow z$, which
suffices for $y\in out_1 (G,a)$.
\end{proof}
SI and WO are two of the three rules that are characteristic of the account of $out_1$ as described in~\cite{mvt2000}. The last one is \\

$
\begin{array}{ll}
  \textrm{(AND)} & \displaystyle\frac{(a,x) \hspace{1cm}
  (a,y)}{(a,x\wedge y)} \\\\
\end{array}
$

The argument for AND is more involved. 
For the sake of clarity we
proceed in two steps, and prove first a qualified form of the rule,
the restriction being that the pairs to which the rule is applied
appear explicitly in $G$.
\begin{Obs}[Restricted AND] Suppose $(a,x), (a,y)\in G$. Assume 
$out_1$ is taken as in definition~\ref{simple-minded}. We have $x,y\in out_1 (G,a)$, but also
$x\wedge y\in out_1 (G,a)$.
\end{Obs}
\begin{proof} We have $x\in out_1 (G,a)$, because $x\in G(\uparrow a)$ and $\inf  G(\uparrow a)\leq x$. The argument for  $x\in out_1 (G,a)$ is similar. We now show $x\wedge y\in out_1 (G,a)$. Clearly
$x\wedge y\in\uparrow\inf \{x,y\}$. Since $(a,x)$ and $(a,y)$ are in $G$, we have
$x,y\in G(\{a\})$, viz $\{x,y\}\subseteq G(\{a\})$. By monotony for $\uparrow\inf$ (\emph{aka} observation
 \ref{properties} (ii)) $x\wedge y\in\uparrow\inf \{x,y\}\subseteq \uparrow\inf  G(\{a\})$ so 
$x\wedge y\in out_1 (G,a)$ as required.
\end{proof}
The argument for the ``restricted AND" rule can be restated informally
as follows. You are in $x$. You jumped there from point $a$. You want
to know if you can move to $x\wedge y$. You go down the ordering until
you can see $y$ as head of another bridge from the same point. If you
can see $y$ that way, the jump is allowed.

Now we show that the rule for AND holds in its plain form. The
argument makes use of Lemmas~\ref{factaboutinf}
ands~\ref{fact2aboutinf}.

\begin{Obs}[Full AND] If $out_1$ is taken as in definition~\ref{simple-minded}, then  $out_1$ validates\\

$
\begin{array}{ll}
  \textrm{(AND)} & \displaystyle\frac{(a,x) \hspace{1cm}
  (a,y)}{(a,x\wedge y)} \\\\
\end{array}
$
\end{Obs}
\begin{proof} Assume $x,y\in out_1 (G,a)$. We want to show $x\wedge y\in out_1 (G,a)$. From the assumption we get that
$G$ contains bridges of the form 
$$(a_1, b_1), ..., (a_n, b_n),(d_1, c_1), ..., (d_m, c_m) $$ such that
\begin{itemize}
\item $a_1,...,a_n,d_1,...,d_m\in \uparrow (a)$ ($a$ is below all the bodies)
\item $\inf (b_1,...,b_n)\leq x$
\item $\inf (c_1,...,c_m)\leq y$
\end{itemize}
By definition $x\wedge y\leq x$ and $x\wedge y\leq y$. 

The crux of the argument consists in showing that $x\wedge y$ is above the infimum of the union of all the heads from which $x$ and $y$ originate, viz
$$\inf (b_1,...,b_n, c_1,...,c_m)\leq x\wedge y$$
This is explained in Figure~\ref{and}. By Lemma~\ref{factaboutinf}, we have
\begin{flalign}
\inf (b_1,...,b_n, c_1,...,c_m)\leq \inf (b_1,...,b_n) \\
\inf (b_1,...,b_n, c_1,...,c_m)\leq \inf (c_1,...,c_m) 
\end{flalign}

\begin{figure}[ht]
     \caption{$x\wedge y$ above the infimum of the union of all the relevant heads}\label{and}
    \begin{center}
    \begin{tikzpicture}
        \tikzstyle{every node} = [rectangle]
        \node (s1) at (5,0) {$\inf (b_1,...,b_n, c_1,...,c_m)$};
        \node (s2) at (10,1) {$\inf (c_1,...,c_m)$};
        \node (s3) at (0,1) {$\inf (b_1,...,b_n)$};
        \node (s4) at (-1,2) {$b_1$};
         \node (s5) at (1,2) {$b_n$};
        \node (s6) at (0,2) {$...$};
        \node (s12) at (0,3) {$x$};
        \node (s7) at (9,2) {$c_1$};
         \node (s8) at (11,2) {$c_m$};
        \node (s9) at (10,2) {$...$};
        \node (s13) at (10,3) {$y$};
        \node (s11) at (5,2) {$x\wedge y$};
        \foreach \from/\to in {s1/s2, s1/s3, s3/s4, s3/s6, s3/s5,s5/s12,s4/s12, s6/s12,s1/s12, s2/s7,s2/s8,s2/s9,s7/s13,s8/s13,s9/s13,s1/s13,s11/s12,s11/s13,s1/s11}
            \draw[-] (\from) -- (\to);
    \end{tikzpicture}
    \end{center}
\end{figure}
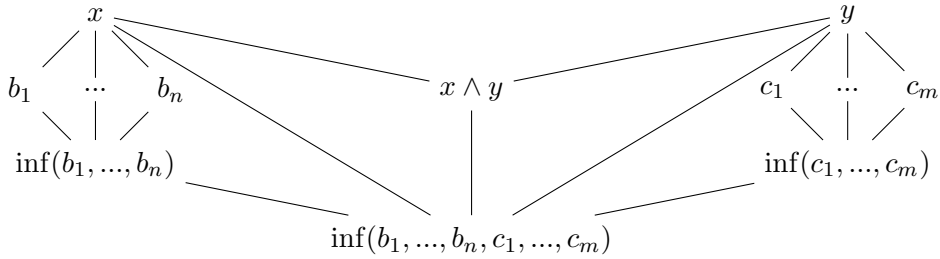
By transitivity,
\begin{flalign}
\inf (b_1,...,b_n, c_1,...,c_m)\leq x\\
\inf (b_1,...,b_n, c_1,...,c_m)\leq y
\end{flalign}
From Lemma~\ref{fact2aboutinf} it follows that
$$\inf (b_1,...,b_n, c_1,...,c_m)\leq x\wedge y$$
With this result established, the remaining of the proof is easy. Since
$$
a_1,...,a_n,d_1,...,d_m\in \uparrow (a)
$$
we have
$$
b_1,...,b_n,c_1,...,c_m\in G(\uparrow (a))
$$
But we have just seen that  $\inf (b_1,...,b_n, c_1,...,c_m)\leq x\wedge y$. So  $x\wedge y\in out_1 (G,a)$
as required.
\end{proof}
$x\in out_1 (G,a)$ and $(a,x)\in out_1 (G,a)$ are two different but
equivalent ways of saying the same thing. This is what makes 
the axiomatic characterization possible. We write $(a, x)\in
deriv_1(G)$, or equivalently $x\in
deriv_1(G, a)$, to indicate that $(a, x)$ is derivable from $G$. 
The formal definition is as in~\cite{mvt2000}. 
A pair
$(a,x)$ is said derivable from $G$ iff it is in the least set that
includes $G$, contains the pair $(1,1)$, and is closed under the
rules of the system (SI), (WO) and (AND).

Makinson and van der Torre make two extra conventions, which will be
adopted here too.  $(A,x)$ is said to be derivable from $G$ iff $(a,
x)$ is derivable from $G$ for some conjunction $a=a_1 \wedge ...\wedge
a_n$ of elements in $A$.  They also understand the conjunction of zero
formulae to be a tautology, so that $(\emptyset,x)$ is derivable from
$G$ iff $(1,x)$.


For finite $A$ and $G$, we have:

\begin{theo}[Soundness]$deriv_1 (G,A)\subseteq out_1 (G,A)$ 
\end{theo}
\begin{proof} Same as in~\cite{mvt2000}.
\end{proof}
\begin{coro}[Completeness] $out_1 (G,A)\subseteq deriv_1 (G,A)$ 
\end{coro}
\begin{proof} 
A re-run of the proof given in~\cite{mvt2000}. 
Assume $x\in \uparrow \inf(G(\uparrow\inf (A))$. We need to show
that $(A,x)$ is derivable from $G$. 
\\

\noindent
\underline{Case 1:} $G(\uparrow\inf A)=\emptyset$. In that case $x$ is $1$.
 But $(1,1)$ is derivable from
$G$. 
For any conjunction $a=a_1 \wedge ...\wedge a_n$ of elements in $A$,
$a\leq 1$. By (SI), $(a,1)$ is derivable from $G$. By definition of derivability, $(A,1)$ is
derivable from $G$ as required. \\

\noindent
\underline{Case 2:} $G(\uparrow\inf A)\not=\emptyset$. 
 So
$G$ contains (finitely many) rules of the form 
$(b_1, x_1), ..., (b_n, x_n)$ such that
\begin{itemize}
\item $b_1,...,b_n\in \uparrow\inf A$ 
\item $\inf (x_1,...,x_n)\leq x$
\end{itemize}
We have $b_1,...,b_n \in \uparrow\inf (a_1,...,a_k)$, where $a_1$, ... $a_k$ are all the elements of $A$.
By definition of derivability,
$(b_1, x_1), ...,$ and $ (b_n, x_n)$ are all derivable from $G$. We
construct a derivation of $(A,x)$ as follows

\begin{prooftree}
          \AxiomC{$(b_1,x_1)$} 
\RightLabel{SI} 	\UnaryInfC{$(\inf (a_1,...,a_k),x_1)$}
	\AxiomC{$\ldots\ldots$} 
\noLine\UnaryInfC{$\;$}
	\AxiomC{$(b_n,x_n)$} 
\RightLabel{SI} 	\UnaryInfC{$(\inf (a_1,...,a_k),x_n)$}
\RightLabel{AND} 	\TrinaryInfC{$(\inf (a_1,...,a_k), x_1\wedge \ldots\wedge x_n)$}
\LeftLabel{WO} \UnaryInfC{$(\inf (a_1,...,a_k), x)$}
	\end{prooftree}
$\inf (a_1,...,a_k)$ is the conjunction of all the elements in
$A$. So from the definition of derivability it follows that $(A, x)$
is derivable.
\end{proof}
The generalization to an input set $A$ and a set $G$ of generators of arbitrary cardinality is left as a topic for future research.

\section{Conclusion and future research}
We have described a geometrical account of I/O logic. A soundness and completeness results was reported  for the simplest I/O operation called simple-minded, and in the finite case only.
We end this note by discussing some topics for future investigation.

\subsection{Other I/O operations}

How to handle the other I/O operations discussed in \cite{mvt2000}?

\subsection{Geometrical defaults}

We can do geometrical defaults in input output logic.  Start with a
default theory $( W, D)$, A default $d$ in $D$ reads: if $a$ is in the
theory and it is consistent that $b$ then let $c$ be in the theory.

$d$ can be read as a tourist travel instruction in the geometry of the
the partial order, as follows: If you are above $a$ and you are not
above $\neg b$ then travel higher in the order until you are also
above $c$. A default extension is a point where you no longer needs to move.
$W$  is your starting point.

This is an informal characterization. One still needs to work out the formal details. 

\subsection{Jump back}

One could work with two sets of jump points, $G_1$ and $G_2$.
Some of the jumps in $G_2$ undo those made in $G_1$ ($G_2$ goes
backwards, but of course it can do more). How to implement this idea?

\subsection{Constraints}

The idea is to combine default reasoning and contrary-to-duty (CTD) reasoning. When there is a
violation, the contrary-to-duty can be given priority over the primary
obligation. This is the standard approach to CTDs in I/O
logics. To avoid inconsistency, we get rid of the primary obligation.
However, the violation can also be seen an exception, in which
case the primary obligation must be given priority. We need a mechanism to
deal with the two views.


\bibliographystyle{plain}
\bibliography{432}

\begin{thebibliography}{1}

\bibitem{cl00}
C.~Cori and D.~Lascar.
\newblock {\em Mathematical logic: Part 1}.
\newblock OUP, Oxford, 2000.

\bibitem{depaiva}
V.~de~Paiva and E.~Ritter.
\newblock Basic constructive modality.
\newblock In {\em Logic without frontiers—Festschrift for Walter Alexandre
  Carnielli on the occasion of his 60th birthday}, pages 411--428. College
  Publication, London, 2011.

\bibitem{LO2000}
L.~Lindahl and J.~Odelstad.
\newblock An algebraic analysis of normative systems.
\newblock {\em Ratio Juris}, 13:261--271, 2000.

\bibitem{mvt2000}
D.~Makinson and L.~van~der Torre.
\newblock Input/output logics.
\newblock {\em Journal of Philosophical Logic}, 29(4):383--408, 2000.

\bibitem{imla17}
X.~Parent.
\newblock A modal translation of an intuitionistic i/o operation, 2017.
\newblock Presented at the 7th Workshop on Intuitionistic Modal Logic and
  Applications (IMLA 2017), organized by V. de Paiva and S. Artemov at the
  University of Toulouse (France), 17-28 July, 2017.

\bibitem{Parent2014}
X.~Parent, D.~Gabbay, and L.~van~der Torre.
\newblock Intuitionistic basis for input/output logic.
\newblock In S.~O. Hansson, editor, {\em David Makinson on Classical Methods
  for Non-Classical Problems}, pages 263--286. Springer Netherlands, Dordrecht,
  2014.

\end{thebibliography}

\end{document}